\documentclass[10pt, onecolumn]{IEEEtran}
\usepackage{amsthm}
\usepackage{amsfonts}
\usepackage{graphicx}
\usepackage{latexsym}
\usepackage{amssymb}
\usepackage{stmaryrd}
\usepackage{comment}
\usepackage{amscd}
\usepackage[small]{caption}
\usepackage[dvipsnames]{xcolor}
\usepackage{algorithm}
\usepackage[noend]{algpseudocode}
\algrenewcommand\alglinenumber[1]{\scriptsize #1:}
\algrenewcommand\algorithmicindent{1em}%
\usepackage{graphicx}
\usepackage{subfigure}
\usepackage{wrapfig}
\usepackage{float}
\usepackage{bbm}
\usepackage{amsmath,bm}
\usepackage{thmtools} 
\usepackage{thm-restate}
\usepackage{caption}

\usepackage{tikz}
\tikzset{
    cross/.pic = {
    \draw[rotate = 45] (-#1,0) -- (#1,0);
    \draw[rotate = 45] (0,-#1) -- (0, #1);
    }
}
\usetikzlibrary{patterns,arrows}
\usetikzlibrary{decorations.pathreplacing,calligraphy}
\usepackage{graphicx} 
\usepackage{booktabs} 
\allowdisplaybreaks

\usepackage{mathtools}
\newcommand{\bea}{\begin{eqnarray}}
\newcommand{\eea}{\end{eqnarray}}
\newcommand{\bean}{\begin{eqnarray*}}
\newcommand{\eean}{\end{eqnarray*}}

\newcommand{\sbinom}[2]{\left( \begin{array}{c} #1 \\ #2 \end{array} \right) }

\newcommand{\XlgA}{\mathsf{X}_\mathsf{g_A}}
\newcommand{\XlgB}{\mathsf{X}_\mathsf{g_B}}
\newcommand{\Xrg}{\mathsf{Y}_\mathsf{g}}
\newcommand{\Xlr}{\mathsf{X}_{s}}

\newcommand{\tx}{\widetilde{\bfx}}

\newcommand{\cE}{{\cal E}}

\newcommand{\cG}{{\cal G}}

\newcommand{\cL}{{\cal L}}
\newcommand{\cM}{{\cal M}}
\newcommand{\cN}{{\cal N}}
\newcommand{\cO}{{\cal O}}
\newcommand{\cP}{{\cal P}}

\newcommand{\cR}{{\cal R}}
\newcommand{\cS}{{\cal S}}
\newcommand{\cT}{{\cal T}}
\newcommand{\cU}{{\cal U}}

\newcommand{\cX}{{\cal X}}
\newcommand{\cY}{{\cal Y}}

\newcommand{\sG}{\script{G}}

\newcommand{\sP}{\script{P}}

\newcommand{\bfa}{{\boldsymbol a}}
\newcommand{\bfb}{{\boldsymbol b}}

\newcommand{\bfd}{{\boldsymbol d}}

\newcommand{\bfx}{{\boldsymbol x}}
\newcommand{\bfy}{{\boldsymbol y}}

\DeclareMathOperator*{\argmin}{arg\,min}

\DeclareMathAlphabet{\mathbfsl}{OT1}{cmr}{bx}{it}
\newcommand{\uuu}{\kern-1pt\mathbfsl{u}\kern-0.5pt}
\newcommand{\vvv}{\kern-1pt\mathbfsl{v}\kern-0.5pt}

\newcommand{\myboxplus}{\kern1pt\mbox{\small$\boxplus$}}

\makeatletter \DeclareRobustCommand{\sbinom}{\genfrac[]\z@{}}
\makeatother
\newcommand{\G}[2]{\sbinom{{#1}\kern-1pt}{{#2}\kern-1pt}}
\newcommand{\Gq}[2]{\sbinom{{#1}\kern-0.25pt}{{#2}\kern-0.25pt}}

\newcommand{\Ps}{\smash{{\sP\kern-2.0pt}_q\kern-0.5pt(n)}}
\newcommand{\sPs}{\smash{{\sP\kern-1.5pt}_q(n)}}
\newcommand{\Ptwo}{\smash{{\sP\kern-2.0pt}_2\kern-0.5pt(n)}}
\newcommand{\Ptwom}{\smash{{\sP\kern-2.0pt}_2\kern-0.5pt(m)}}
\newcommand{\Ptwonm}{\smash{{\sP\kern-2.0pt}_2\kern-0.5pt(n+m)}}
\newcommand{\Ptwoa}{\smash{{\sP\kern-2.0pt}_2\kern-0.5pt(1)}}
\newcommand{\Ptwob}{\smash{{\sP\kern-2.0pt}_2\kern-0.5pt(2)}}
\newcommand{\Ptwoc}{\smash{{\sP\kern-2.0pt}_2\kern-0.5pt(3)}}
\newcommand{\Ptwod}{\smash{{\sP\kern-2.0pt}_2\kern-0.5pt(4)}}
\newcommand{\Ptwoe}{\smash{{\sP\kern-2.0pt}_2\kern-0.5pt(5)}}
\newcommand{\Ptwof}{\smash{{\sP\kern-2.0pt}_2\kern-0.5pt(6)}}
\newcommand{\Ptwokm}{\smash{{\sP\kern-2.0pt}_2\kern-0.5pt(2k-1)}}
\newcommand{\Pone}{\smash{{\sP\kern-2.5pt}_2\kern-0.5pt(n{-}1)}}

\newcommand{\Gr}{\smash{{\sG\kern-1.5pt}_q\kern-0.5pt(n,k)}}
\newcommand{\Gi}{\smash{{\sG\kern-1.5pt}_q\kern-0.5pt(n,i)}}
\newcommand{\Gj}{\smash{{\sG\kern-1.5pt}_q\kern-0.5pt(n,j)}}
\newcommand{\Grmk}{\smash{{\sG\kern-1.5pt}_q\kern-0.5pt(n,n-k)}}
\newcommand{\Grdk}{\smash{{\sG\kern-1.5pt}_q\kern-0.5pt(2k,k)}}
\newcommand{\Grekappa}{\smash{{\sG\kern-1.5pt}_q\kern-0.5pt(n,e+1-\kappa)}}
\newcommand{\Grtwoekappa}{\smash{{\sG\kern-1.5pt}_q\kern-0.5pt(n,2e+1-\kappa)}}
\newcommand{\Gremkappa}{\smash{{\sG\kern-1.5pt}_q\kern-0.5pt(n,e-\kappa)}}
\newcommand{\Gn}{\smash{{\sG\kern-1.5pt}_2\kern-0.5pt(n,n{-}1)}}
\newcommand{\Gnq}{\smash{{\sG\kern-1.5pt}_q\kern-0.5pt(n,n{-}1)}}
\newcommand{\Gone}{\smash{{\sG\kern-1.5pt}_2\kern-0.5pt(n,1)}}
\newcommand{\Gqone}{\smash{{\sG\kern-1.5pt}_q\kern-0.5pt(n,1)}}
\newcommand{\GTwo}{\smash{{\sG\kern-1.5pt}_2\kern-0.5pt(n,k)}}
\newcommand{\GTwonk}[2]{{\smash{{\sG\kern-1.5pt}_2\kern-0.5pt({#1},{#2})}}}
\newcommand{\Gnk}{\smash{{\sG\kern-1.5pt}_2\kern-0.5pt(n,n{-}k)}}
\newcommand{\Greone}{\smash{{\sG\kern-1.5pt}_q\kern-0.5pt(n,e{+}1)}}
\newcommand{\Gretwo}{\smash{{\sG\kern-1.5pt}_q\kern-0.5pt(n,e{+}2)}}

\newcommand{\be}[1]{\begin{equation}\label{#1}}
\newcommand{\ee}{\end{equation}}

\newcommand{\ab}[1]{{\footnotesize  [{\textcolor{blue}{#1}} \textcolor{blue!60!black}{--avital}]\normalsize}}

\newcommand{\Cref}[1]{Co\-rol\-la\-ry\,\ref{#1}}

\newtheorem{lemma}{Lemma}

\newtheorem{corollary}{Corollary}

\newtheorem{definition}{Definition}
\newtheorem{proposition}{Proposition}

\newtheorem{problem}{Problem}

\newcommand{\hanmao}[1]{{\footnotesize [\gcomment{#1}\;\;\vcomment{--HanMao}]}}
\newcommand{\vcomment}[1]{{\color{violet}#1}}
\newcommand{\gcomment}[1]{{\color{OliveGreen}#1}}

\makeatletter
\newenvironment{breakablealgorithm}
  {
   \begin{center}
     \refstepcounter{algorithm}
     \hrule height.8pt depth0pt \kern2pt
     \renewcommand{\caption}[2][\relax]{
       {\raggedright\textbf{\fname@algorithm~\thealgorithm} ##2\par}%
       \ifx\relax##1\relax 
         \addcontentsline{loa}{algorithm}{\protect\numberline{\thealgorithm}##2}%
       \else 
         \addcontentsline{loa}{algorithm}{\protect\numberline{\thealgorithm}##1}%
       \fi
       \kern2pt\hrule\kern2pt
     }
  }{
     \kern2pt\hrule\relax
   \end{center}
  }
\makeatother

\IEEEoverridecommandlockouts

\begin{document}

\author{%
\small
\IEEEauthorblockN{\textbf{Shubhransh~Singhvi}\IEEEauthorrefmark{1}, 
\textbf{Avital~Boruchovsky}\IEEEauthorrefmark{2}, 
\textbf{Han~Mao~Kiah}\IEEEauthorrefmark{3}
and \textbf{Eitan~Yaakobi}\IEEEauthorrefmark{2}}\\
  \IEEEauthorblockA{\IEEEauthorrefmark{1}%
                      Signal Processing  \&  Communications Research  Center, International Institute  of  Information Technology, Hyderabad, India}\\
\IEEEauthorblockA{\IEEEauthorrefmark{2}%
                     Department of Computer Science, 
                     Technion---Israel Institute of Technology, 
                     Haifa 3200003, Israel}\\
  \IEEEauthorblockA{\IEEEauthorrefmark{3}%
                     School of Physical and Mathematical Sciences, 
		Nanyang Technological University, Singapore
         }
 }

\title{\textbf{Data-Driven Bee Identification for DNA Strands}}
\date{\today}
\maketitle
\thispagestyle{empty}	
\pagestyle{empty}


\begin{abstract}
     We study a data-driven approach to the bee identification problem for DNA strands. 
    The bee-identification problem, introduced by Tandon et al. (2019), requires one to identify $M$ bees, each tagged by a unique barcode, via a set of $M$ noisy measurements. 
    Later, Chrisnata et al. (2022) extended the model to case where one observes $N$ noisy measurements of each bee, and applied the model to address the unordered nature of DNA storage systems.
    
    In such systems, a unique address is typically prepended to each DNA data block to form a DNA strand, but the address may possibly be corrupted.
    While clustering is usually used to identify the address of a DNA strand,
    this requires $\cM^2$ data comparisons (when $\cM$ is the number of reads).
    In contrast, the approach of Chrisnata et al. (2022) avoids data comparisons completely. In this work, we study an intermediate, data-driven approach to this identification task. 

    For the binary erasure channel,  we first show that we can almost surely correctly identify all DNA strands under certain mild assumptions. 
    Then we propose a data-driven pruning procedure and demonstrate that on average the procedure uses only a fraction of $\cM^2$ data comparisons.
    Specifically, for $\cM = 2^n$ and erasure probability $p$, the
    expected number of data comparisons performed by the procedure is $\kappa\cM^2$, where $\left(\frac{1+2p-p^2}{2}\right)^n \leq \kappa \leq \left(\frac{1+p}{2}\right)^n $. 

\end{abstract}

    
\section{Introduction}
Existing storage technologies cannot keep up with the modern data explosion. Current solutions for storing huge amounts of data uses magnetic and optical disks. Despite improvements in optical discs, storing a zettabyte of data would still take many millions of units, and use significant physical space. Certainly, there is a growing need for a significantly more durable and compact storage system. The  potential of macromolecules in ultra-dense storage systems was recognized as early  as in the 1960s, when  the  celebrated physicists Richard Feynman outlined his vision for nanotechnology in the talk `There is plenty of room at the bottom'. Using DNA is an attractive possibility because it is extremely dense (up to about 1 exabyte per cubic millimeter) and durable (half-life of over 500 years).
Since the first experiments conducted by Church et al. in 2012~\cite{Church.etal:2012} and Goldman et al. in 2013~\cite{Goldman.etal:2013}, there have been a flurry of experimental demonstrations (see \cite{Shomorony.2022,Yazdi.etal:2015b} for a
survey). Amongst the various coding design considerations, in this work, we study the unsorted nature of the DNA storage system~\cite{LSWY18, Shomorony.2022}.

A DNA storage system consists of three important components. The first is the DNA synthesis which produces the oligonucleotides, also called \textit{strands}, that encode the data. The second part is a storage container with compartments  which  stores the  DNA  strands, however without order. Finally, to retrieve the data, the DNA is accessed using next-generation sequencing, which results in several noisy copies, called \textit{reads}. The processes of synthesizing, storing, sequencing, and handling strands are all error prone. Due to this unordered nature of DNA-based storage systems, when the user retrieves the information, in addition to decoding the data, the user has to determine the identity of the data stored in each strand. A typical solution is to simply have a set of addresses and store this address information as a prefix to each DNA strand. As the addresses are also known to the user, the user can identify the information after the decoding process. As these addresses along with the stored data are prone to errors, this solution needs further refinements.

In \cite{Organick}, the strands (strand = address + data) are first clustered with respect to the edit distance. 
Then the authors determine a consensus output amongst the strands in each cluster and 
finally, decode these consensus outputs using a classic concatenation scheme.  
For this approach,  the clustering step is computationally expensive.
When there are $\cM$ reads, the usual clustering method involves $\cM^2$ pairwise comparisons to compute distances. This is costly when the data strands are long, and 
the problem is further exacerbated if the metric is the edit distance.
Therefore, in  \cite{clustering}, a distributed approximate clustering algorithm was proposed and the authors clustered 5 billion strands in 46 minutes on 24 processors.

In \cite{Chrisnata2022}, the authors proposed an approach that avoids clustering.
Informally, the bee-identification problem requires the receiver to identify $M$ “bees” using a set of $M$ unordered noisy measurements~\cite{TTV2019}. 
Later, in~\cite{Chrisnata2022}, the authors generalized the setup to multi-draw channels where every bee (strand) results in $N$ noisy outputs (reads). 
The task then is to identify each of the $M$ bees from the $MN$ noisy outputs and it turns out that this task can be reduced to a minimum-cost network flow problem.
In contrast to previous works, the approach in~\cite{Chrisnata2022} utilizes only the noisy addresses, which are of significantly shorter length, and the method does not take into account the associated noisy data. 
Hence, this approach involves no data comparisons.

In this work, we consider an intermediate, data-driven approach to the identification task by drawing ideas from the clustering and the bee identification problems. Specifically, we focus on the case of the binary erasure channel and the case where the addresses are uncoded. 
We first show that we can almost surely correctly identify all DNA strands under certain mild assumptions. 
Then we propose a data-driven pruning procedure and demonstrate that on average the procedure uses only a fraction of $\cM^2$ data comparisons (when there are $\cM$ reads). 
We formally define our problem in the next section.

\section {Problem Formulation}
\label{sec:prob}
Let $N$ and $M$ be positive integers. Let $[M]$ denote the set $\{1,2,\ldots,M\}$. An $N$-permutation $\psi$ over $[M]$ is an $NM$-tuple $(\psi(j))_{j\in[MN]}$ where every symbol in $[M]$ appears exactly $N$ times, and we denote the set of all $N$-permutations over $[M]$ by $\mathbb{S}_N(M)$.
Let $C \subseteq{\{0,1\}}^n$ be a binary code of length $n$ and size $M$ (the addresses), and assume that every codeword $\bfx_i \in C$ is attached to a length-$L$ data part $\bfd_i\in {\{0,1\}}^L$ to form a strand, which is the tuple, $(\bfx_i,\bfd_i)$. We denote the ratio of the length of the data part to the length of the address part by $\beta$, i.e., $L \triangleq \beta n$, where $\beta \in \mathbb{R_{+}}$. Let the multiset of  data be denoted by $D = \{\{\bfd_i:i\in[M]\}\}$ and the set of strands by $R=\{(\bfx_i,\bfd_i):i\in[M]\}$. Throughout this paper, we assume that $D$ is drawn uniformly 
at random over $\{0,1\}^L$ and that $C$ is the whole space. Let $\cS_N((\bfx,\bfd))$ denote the multiset of channel outputs when $(\bfx,\bfd)$ is transmitted $N$ times through the channel $\cS$. Assume that the entire set $R$
is transmitted through the channel $\cS$, hence an unordered multiset, $R'_{N}=\{\{(\bfy_1,\bfd_1'),(\bfy_2,\bfd_2'),\ldots,(\bfy_{MN},\bfd_{MN}')\}\}$, of $MN$ noisy strands (reads)  is obtained, where for every $j\in[MN]$, $(\bfy_j, \bfd_j') \in \cS_N((\bfx_{\pi(j)},\bfd_{\pi(j)}))$ for some  $N$-permutation $\pi$ over $[M]$, which will be referred to as the \textit{true $N$-permutation}. Note that the receiver, apart from the set of reads $R'_N$, has access to the set of addresses $C$ but does not know the set of data $D$. In this work, we first consider the following problem.

\begin{problem}
\label{prob1}
Given $\cS$ and $\epsilon$, find the region $\cR\in\mathbb{R}_+^2$, such that for $(N,\beta) \in \cR$,  it is possible to identify the true permutation with probability at least $1-\epsilon$ when the code $C$ is the whole space and the data $D$ is drawn uniformly at random. 
\end{problem}

For $(N,\beta)\in\cR$, we can find the true permutation by making at least $(N|C|)^2$ data comparisons. This may be expensive when the data parts are long, i.e., when $L$ is large. Therefore, our second objective is to reduce the number of data comparisons.

\begin{problem}
\label{prob2}
   Let $\kappa<1$. Given $\cS$ and $\epsilon$ and $(N,\beta)\in \cR$, design an algorithm to identify the true permutation with probability at least $1-\epsilon$ using $\kappa (N|C|)^2$ data comparisons. As before, $C$ is the whole space and $D$ is  drawn uniformly at random. 
\end{problem}



Unless otherwise stated, we assume that $\cS$ is the $\mathsf{BEC}(p)$ channel with $0<p<1$.
In Section~\ref{Sec:PMA}, we first propose an extension of the Peeling Matching Algorithm~\cite{Kiah2021} to the multi-draw erasure channel. We then demonstrate that the peeling matching algorithm identifies the true permutation with a vanishing probability as $n$ grows. In Section~\ref{sec:UniqPerm}, we address Problem~\ref{prob1} and identify the region $\cR$ for which there exists only one valid permutation, viz. the true permutation. In Section~\ref{Sec:DPA}, we propose a data-driven pruning algorithm that identifies the true-permutation with probability at least $1-\epsilon$ when $(N,\beta) \in \cR$. In Section~\ref{Sec:DPA_analysis}, we analyse the expected number of data comparisons performed by the data-driven pruning algorithm.

\section{The Peeling Matching Algorithm ($\mathsf{PMA}$)}
\label{Sec:PMA}
In this section, we extend the Peeling Matching Algorithm (PMA), presented in \cite{Kiah2021} for $N=1$ to a general value of $N$. The $\mathsf{PMA}$-based approach uses solely the information stored in the addresses to identify the true-permutation, and does not take into consideration the noisy data that is also available to the receiver.  The first step in the peeling matching algorithm is to construct a bipartite undirected graph $\cG = (\cX\cup \cY, E)$, where the left nodes are the addresses ($\cX=C$) and the right nodes are the noisy reads ($\cY = R'_N$). There exists an edge between $\bfx\in \cX$ and $(\bfy,\bfd')\in \cY$ if and only if $P(\bfy|\bfx)>0$, where $P(\bfy|\bfx)$ is the likelihood probability of observing $\bfy$ given that $\bfx$ was transmitted. For $\bfx\in \cX$ and $(\bfy,\bfd')\in\cY$ let $E_\bfx$
and $E_{(\bfy,\bfd')}$ denote the multiset of neighbours of $\bfx$ and the set of neighbours of $(\bfy,\bfd')$ in $\cG$, respectively, i.e.,  $E_\bfx=\{\{(\bfy,\bfd')|(\bfx,(\bfy,\bfd'))\in E\}\}$, $E_{(\bfy,\bfd')}=\{\bfx|(\bfx,(\bfy,\bfd'))\in E\}$. Note that the degree of every left node is at least $N$ as $\cS_N((\bfx,\bfd)) \subseteq{} E_\bfx$. For right nodes and left nodes with degrees $1$ and $N$ respectively, the corresponding neighbor(s) can be matched with certainty. For ease of exposition, we refer to such nodes as \textit{good} nodes. 
\begin{definition}
A node $(\bfy,\bfd')\in \cY$ is said to be a \textbf{good right node} if $|E_{(\bfy,\bfd')}| = 1$. A node $\bfx \in \cX$ is said to be a \textbf{Type-A good left node} if $|E_\bfx| = N$ or a \textbf{Type-B good left node} if $|\{(\bfy,\bfd')\in E_x: |E_{(\bfy,\bfd')}| =1\}| = N$. 
\end{definition}

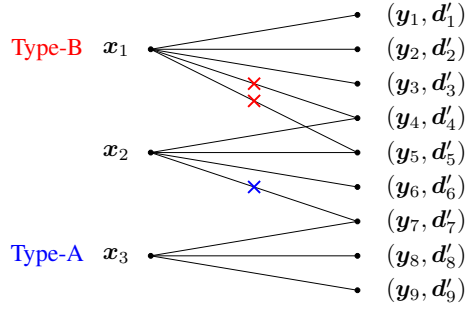
\begin{figure}[H]
\centering
\resizebox{6.5cm}{!}
{
\begin{tikzpicture}
\filldraw (1.5,0.5) circle (1pt);
\filldraw (1.5,1) circle (1pt);
\filldraw (1.5,1.5) circle (1pt);
\filldraw (1.5,2) circle (1pt);
\filldraw (1.5,2.5) circle (1pt);
\filldraw (1.5,3) circle (1pt);
\filldraw (1.5,3.5) circle (1pt);
\filldraw (1.5,4) circle (1pt);
\filldraw (1.5,4.5) circle (1pt);

\draw (2.5,0.5) node{$(\bfy_9,\bfd'_9)$};
\draw (2.5,1) node{$(\bfy_8,\bfd'_8)$};
\draw (2.5,1.5) node{$(\bfy_7,\bfd'_7)$};
\draw (2.5,2) node{$(\bfy_6,\bfd'_6)$};
\draw (2.5,2.5) node{$(\bfy_5,\bfd'_5)$};
\draw (2.5,3) node{$(\bfy_4,\bfd'_4)$};
\draw (2.5,3.5) node{$(\bfy_3,\bfd'_3)$};
\draw (2.5,4) node{$(\bfy_2,\bfd'_2)$};
\draw (2.5,4.5) node{$(\bfy_1,\bfd'_1)$};

\filldraw (-1.5,1) circle (1pt);
\filldraw (-1.5,2.5) circle (1pt);
\filldraw (-1.5,4) circle (1pt);

\draw (-1.5,1) -- (1.5,0.5);
\draw (-1.5,1) -- (1.5,1);
\draw (-1.5,1) -- (1.5,1.5);

\draw (-1.5,2.5) -- (1.5,1.5);
\draw (-1.5,2.5) -- (1.5,2);
\draw (-1.5,2.5) -- (1.5,2.5);
\draw (-1.5,2.5) -- (1.5,3);

\draw (-1.5,4) -- (1.5,2.5);
\draw (-1.5,4) -- (1.5,3);
\draw (-1.5,4) -- (1.5,3.5);
\draw (-1.5,4) -- (1.5,4);
\draw (-1.5,4) -- (1.5,4.5);

\draw (-2,4) node{$\bfx_1$};
\draw (-2,2.5) node{$\bfx_2$};
\draw (-2,1) node{$\bfx_3$};

\draw (-3,4) node{\textcolor{red}{Type-B}};
\draw (-3,1) node{\textcolor{blue}{Type-A}};
\path (0,2) pic[blue, rotate = 0, thick] {cross=3.5pt};
\path (0,3.5) pic[red, rotate = 0, thick] {cross=3.5pt};
\path (0,3.25) pic[red, rotate = 0, thick] {cross=3.5pt};


\end{tikzpicture}
}
\caption{Let $N=3$. If $\bfx_1$ is peeled, then $\bfx_2$ becomes a Type-B good left node, and if  $\bfx_3$ is peeled then $\bfx_2$ becomes a Type-A good left~node.}
\label{fig:PMA}
\end{figure}

Let $\Xrg, \XlgA$ and $\XlgB$ denote the set of good right nodes, Type-A good left nodes and Type-B good left nodes, respectively. The peeling matching algorithm when executed over $\cG$, finds good left nodes and identifies the corresponding $N$ channel outputs until there are no good left nodes. Let $\cP_\cG = \left(\cX \cup \cY, \cP_E\right)$ denote the bipartite matching identified by the Peeling Matching Algorithm. 

\begin{breakablealgorithm}
\caption{Peeling Matching Algorithm}
\label{Alg1}
\begin{algorithmic}[1]
\Procedure{Peel}{$\cP_\cG,\cG,\bfx$}
\If{$\bfx\in\XlgA$}
\For{$(\bfy,\bfd')\in E_\bfx$}
    \State Remove $\{(\tx,(\bfy,\bfd')): \tx\in E_{(\bfy,\bfd')}\}$ from $E$
    \State Add $(\bfx, (\bfy,\bfd'))$ to $\cP_E$
    \State Remove $(\bfy,\bfd')$ from $\cY$  
\EndFor
\State Remove $\bfx$ from $\cX$

\ElsIf{$\bfx\in\XlgB$}
    \State Add $\{(\bfx,(\bfy,\bfd')):(\bfy,\bfd') \in E_x\cap \Xrg\}$ to $\cP_E$
   \State Remove $\{(\bfx,(\bfy,\bfd')):(\bfy,\bfd') \in E_x\}$ from $E$ and remove $E_x \cap Y_g$ from $\cY$
   \State Remove $\bfx$ from $\cX$
\EndIf
\EndProcedure
\Procedure{PMA}{$\cP_\cG,\cG$}
\For{$\bfx\in \XlgA \cup \XlgB$}
    \State PEEL$(\cP_\cG,\cG,\bfx)$
\EndFor
\If {$|\cP_E| = N2^n$}
\State \Return $\cP_\cG$
\Else 
\State \Return $\mathsf{FAILURE}$
\EndIf
\EndProcedure
\end{algorithmic}
\end{breakablealgorithm}
Note that as shown in Fig.~\ref{fig:PMA}, peeling Type-A and Type-B good left nodes might generate new Type-A and Type-B good left nodes, respectively. Thus, at any instant during the course of the algorithm, we assume $\XlgA, \XlgB$ to reflect the Type-A and Type-B good left nodes, respectively, at that instant. 

\begin{proposition}\cite{KVY2022}
\label{cycle}
    Algorithm~\ref{Alg1} finds the true permutation if only if there are no cycles in $\cG$. 
\end{proposition}




 Let the multiset of right nodes that are in a cycle be denoted by $\mathsf{Y_{cycle}}$. In the next lemma, we derive a lower bound on the probability of observing at least one cycle in $\cG$. 
\begin{restatable}{lemma}{CyclesPMA}\label{PMA}
The probability of observing at least one cycle in $\cG$ is lower bounded by
    \begin{align*}
    P(|\mathsf{Y_{cycle}}|>1)&>1-\frac{\cU_\mathsf{cycle}}{N2^n(1-\cU_\mathsf{cycle})},
\end{align*}
where $\cU_\mathsf{cycle} \triangleq 2^{-N\big((1+p^2)^n-1\big)}$.
\end{restatable}
\begin{proof}
Let $\mathsf{Y^*_{cycle}}$ denote the multiset of right nodes that are in a cycle of size four. Let $(\bfy,\bfd') \in \cS_N((\bfx,\bfd))$ and $ (\widetilde{\bfy},\widetilde{\bfd'}) \in \cS_N((\tx,\widetilde{\bfd}))$ such that $d_H(\bfx,\tx) = r>0$. It can be verified that
$\{\bfx,(\bfy,\bfd'),\tx, (\widetilde{\bfy},\widetilde{\bfd'})\}$ forms a cycle with probability $p^{2r}$. Therefore, the probability that $(\bfy,\bfd')$ is not in a cycle is 
\begin{align*}
P\left((\bfy,\bfd')\notin\mathsf{Y^*_{cycle}}\right) &= \prod_{r=1}^{n} \left(1-p^{2r}\right)^{N\binom{n}{r}}\\
 &=2^{\log\left(\prod_{r=1}^{n}\left(1-p^{2r}\right)^{N\binom{n}{r}}\right)}\\
&=2^{N\sum_{r=1}^{n} \left({\binom{n}{r}}\log\left(1-p^{2r}\right)\right).}
\intertext{Using Jensen's inequality, it can be verified that}
P\left((\bfy,\bfd')\notin\mathsf{Y^*_{cycle}}\right) &\leq 2^{N(2^n-1)\log\left(\frac{\sum_{r=1}^{n}\binom{n}{r}(1-p^{2r})}{2^n-1}\right)} \\
&= 2^{N(2^n-1)\log\left(1-\frac{(1+p^2)^n-1}{2^n-1}\right)} .\\
\intertext{Using Taylor series expansion, for $|a| < 1,~ \log(1+a) \leq a$. Therefore, we have that}
P\left((\bfy,\bfd')\notin\mathsf{Y^*_{cycle}}\right) &\leq \cU_\mathsf{cycle} \triangleq 2^{-N\big((1+p^2)^n-1\big)}. 
\end{align*}
From linearity of expectation, we have that
\begin{align*}
  &\mathbb{E}\left[|\mathsf{Y^*_{cycle}}|\right] = \sum_{(\bfy,\bfd')\in \cY}\mathbb{E}[\mathbb{I}_{\{(\bfy,\bfd')\in\mathsf{Y^*_{cycle}}\}}]\\
  &= N2^n\left(1-\prod_{r=1}^{n} \left(1-p^{2r}\right)^{N\binom{n}{r}}\right).
\end{align*}
Further, from the linearity of variances of indicator random variables, we have that
\begin{align*}
  &\mathsf{Var}\left[|\mathsf{Y^*_{cycle}}|\right] = \sum_{(\bfy,\bfd')\in \cY}\mathsf{Var}[\mathbb{I}_{\{(\bfy,\bfd')\in\mathsf{Y^*_{cycle}}\}}]\\
  &= N2^n\left(1-\prod_{r=1}^{n} \left(1-p^{2r}\right)^{N\binom{n}{r}}\right)\left(\prod_{r=1}^{n} \left(1-p^{2r}\right)^{N\binom{n}{r}}\right).
\end{align*}
Next, from Chebyshev's inequality and the upper bound on $P\left((\bfy,\bfd')\notin\mathsf{Y_{cycle}}\right)$, it can be verified that  
\begin{align*}
    P(|\mathsf{Y^*_{cycle}}|<1)&<\frac{\mathsf{Var}\left[|\mathsf{Y^*_{cycle}}|\right]}{(\mathbb{E}\left[|\mathsf{Y^*_{cycle}}|\right])^2}<\frac{\cU_\mathsf{cycle}}{N2^n(1-\cU_\mathsf{cycle})}.
\end{align*}
Therefore, the results follows. 
\end{proof}

Hence, from Proposition~\ref{cycle} and Lemma~\ref{PMA}, it is highly improbable (vanishingly low probability) to find the true permutation using only the addresses and its noisy measurements (i.e., $\bfx$ and $\bfy$'s). 
In the next section, we see that by making use of the data parts, we can find the true permutation under certain mild assumptions.

\section{Uniqueness of the $N$-Permutation}
\label{sec:UniqPerm}

In this section, we study Problem \ref{prob1} when $\cS =\mathsf{BEC}(p)$. 
Specifically, in Lemmas~\ref{b_th} and~\ref{N_th}, we determine the values $\beta_{\mathsf{Th}}$ and $N_{\mathsf{Th}}$, respectively, such that for all $\beta\geq\beta_\mathsf{Th}$ and $N\geq N_{\mathsf{Th}}$, we are able to find the true permutation with high probability. The result is formally stated in Theorem~\ref{thm:prob1}.

For $\cS =\mathsf{BEC}(p)$, the task of identifying the true permutation $\pi$, can be split into two steps. We can first identify the \textit{partitioning} $\{\cS_N((\bfx_i,\bfd_i)):i\in[M]\}$ and then for each \textit{partition} ($\cS_N((\bfx_i,\bfd_i))$) identify the \textit{label}, viz. the channel input ($\bfx_i$), where $i \in [M]$. Hence, given $R_N' $ and $C$, we are able to find the true permutation if and only if there exists only one valid partitioning and one valid labelling.

Before formally defining partitioning and labelling, we introduce some notations. Let $\bfa_1, \bfa_2\in\{0,1\}^\ell$. For $i\in\{1,2\}$, let $\bfb_i \in\{0,1,*\}^\ell$ be the output of $\bfa_i$ through $\mathsf{BEC}(p)$. We denote the event of $\bfb_1$ and $\bfb_2$ agreeing at the non-erased positions by $\bfb_1 \cong \bfb_2$. For example, let $\bfa_1 = \texttt{00000}, \bfa_2 = \texttt{00011}$ and let $\bfb_1 = \texttt{0000*}, \bfa_2 = \texttt{000*1}$ then $\bfb_1 \Tilde{=} \bfb_2$. Furhter, by abuse of notation, we would denote the event of all sequences in $A\subseteq\cS_N(\bfa_1)$ agreeing at the non-erased positions with all sequences in $B\subseteq\cS_N(\bfa_2)$ by $A\cong B$, where $\cS_N(\bfa_i)$ denotes the multiset of channel outputs when $\bfa_i$ is transmitted $N$ times through the channel $\cS$, $i\in{1,2}$. A right node $(\bfy, \bfd') \in \cY$ is said to be \textbf{faulty} if there exists $(
\widetilde{\bfy},\widetilde{\bfd'})\in\cY\backslash\{(\bfy,\bfd')\}$  with $(\bfy, \bfd')\in\cS_N((\bfx,\bfd))$ and $(\widetilde{\bfy},\widetilde{\bfd'})\in\cS_N((\tx,\widetilde{\bfd}))$ such that $(\bfy, \bfd') \cong (\widetilde{\bfy},\widetilde{\bfd'})$. Let $\mathsf{Y_{faulty}}$ denote the multiset of such faulty nodes. In the next lemma, we calculate the probability of right node being faulty.

\begin{restatable}{lemma}{ProbRightFaulty}\label{right_faulty}

For $(\bfy,\bfd')\in\cY$, $P((\bfy,\bfd')\in \mathsf{Y_{faulty}})$ is 
\begin{small}
\begin{align*}
 1- \prod_{r=1}^{n} \left(1-(2p-p^2)^{r}\left(1-\frac{1}{2}(1-p)^2\right)^L\right)^{N\binom{n}{r}}.  
\end{align*}
\end{small}
\end{restatable}
\begin{proof}    
     For $(\bfx,\bfd), (\tx,\widetilde{\bfd}) \in \cX$, let $(\bfy,\bfd')\in\cS_N((\bfx,\bfd))$ and $(\widetilde{\bfy},\widetilde{\bfd'}) \in \cS_N((\tx,\widetilde{\bfd}))$ with $d_H(\bfx,\tx)=r>0$. For $\bfy\cong\widetilde{\bfy}$, the positions where $\bfx$ and $\tx$ differ must be erased in at least one of them, which happens with probability $(1-(1-p)^2)^r=(2p-p^2)^r$. For index $i \in [L]$, the probability that both $\bfd'$ and $\widetilde{\bfd'}$ are not erased and disagree on $i$ is $\frac{1}{2}(1-p)^2$. Therefore, $P(\bfd' \cong \widetilde{\bfd'}) = \left(1-\frac{1}{2}(1-p)^2\right)^L$. Hence, $P \left((\bfy,\bfd') \cong(\widetilde{\bfy},\widetilde{\bfd'})\right) =(2p-p^2)^{r}\left(1-\frac{1}{2}(1-p)^2\right)^L$. Thus,
\begin{small}
\begin{align*}
P(\mathbb{I}_{(\bfy,\bfd')\not\in \mathsf{Y_{faulty}}}) = \prod_{r=1}^{n} \left(1-(2p-p^2)^{r}\left(1-\frac{1}{2}(1-p)^2\right)^L\right)^{N\binom{n}{r}}.  
\end{align*}
\end{small}
\end{proof}

\begin{definition}
     A \textbf{partitioning} $\cP=\{P_1, P_2, \ldots, P_M\}$ of $\cY$ is defined as the collection of disjoint submultisets of $\cY$, each of size $N$, such that for $i\in[M]$, for $(j,k)\in\binom{[N]}{2}$, $(\bfy_j, \bfd'_j) \cong (\bfy_k, \bfd'_k)$, where $(\bfy_j, \bfd'_j), (\bfy_k, \bfd'_k) \in P_i$. 
\end{definition}

We will refer to $\cP^* \triangleq \{\cS_N((\bfx_i,\bfd_i)):i\in[M]\}$ as the \textit{true partitioning} of $\cY$. Let $\mathbb{P}_\cY$ denote the set of all possible partitionings of $\cY$. Note that if $\vert\mathbb{P}_\cY\vert = 1$ then $\mathbb{P}_\cY = \{\cP^*\}$. Let $\cG' = (\cY, E')$. Now consider the graph, $\cG' = (\cX,E')$, where $\cY = R_N'$. For $(\bfy,\bfd'),(\Tilde{\bfy},\Tilde{\bfd'}) \in\cY, ((\bfy,\bfd'),(\Tilde{\bfy},\Tilde{\bfd'})\in E'$ if $(\bfy,\bfd') \Tilde{=}(\Tilde{\bfy},\Tilde{\bfd'})$.
Note that a partitioning $\cP \in \mathbb{P}_\cY$ corresponds to partitioning the graph $\cG'$ into $M$ cliques each of size $N$. 

\begin{proposition}
    $\vert\mathbb{P}_\cY\vert=1$ if and only if there exists a unique partitioning of the graph $\cG'$ into $M$ cliques each of size $N$. 
\end{proposition}

In the next lemma, we derive a threshold on $\beta$ such that for $\beta\geq\beta_{\mathsf{Th}}, \mathbb{P}_\cY = \{\cP^*\}$ with probability at least $1-\epsilon_1$.

\begin{restatable}{lemma}{BoundBeta}\label{b_th}

For $\beta\geq\beta_{\mathsf{Th}} \triangleq \frac{\log_{2}\left(\frac{N((1+2p-p^2)^n-1)}{\sqrt[N]{\epsilon_1/2^n}}\right)}{n(1-\log_2(1+2p-p^2))}$, we have that $\mathbb{P}_\cY = \{\cP^*\}$ with probability at least $1-\epsilon_1$.
\end{restatable}

\begin{proof}    
Note that for every $\bfx\in\cX$, if there exists at least one $(\bfy,\bfd')\in\cS_N(\bfx,\bfd)$ such that $(\bfy,\bfd')\not\cong\cY\backslash\cS_N(\bfx,\bfd)$ then the only valid partitioning is $\cP^*$. Let $\mathsf{X_{faulty}}$ denote the set of left nodes with $\cS_N(\bfx, \bfd)\subset\mathsf{Y_{faulty}}$. 
From Markov Inequality, 
\begin{align*}
     P(\bfx \in \mathsf{X_{faulty}}) 
     & = P(\mathbb{I}_{\cS_N(\bfx,\bfd)\subset \mathsf{Y_{faulty}}}\geq1)\\
     & \leq\mathbb{E}\left[\mathbb{I}_{\cS_N(\bfx,\bfd)\subset \mathsf{Y_{faulty}}}\right] 
     = \left(P(\mathbb{I}_{(\bfy,\bfd')\subset \mathsf{Y_{faulty}}})\right)^N.
\end{align*}
Therefore, from Lemma \ref{right_faulty}, $P(\bfx \in \mathsf{X_{faulty}})$ is at most 
\begin{small}
\begin{align*}
   \left(1-\prod_{r=1}^{n} \left(1-(2p-p^2)^{r}\left(1-\frac{1}{2}(1-p)^2\right)^L\right)^{N\binom{n}{r}}\right)^N.
\end{align*}
\end{small}
From linearity of expectation, $\mathbb{E}\left[|\mathsf{X_{faulty}}|\right]$ is at most
\begin{small}
\begin{align*}
2^n \left(1-\prod_{r=1}^{n} \left(1-(2p-p^2)^{r}\left(1-\frac{1}{2}(1-p)^2\right)^L\right)^{N\binom{n}{r}}\right)^N.
\end{align*}
\end{small}
Using Weierstrass inequality, we have that  
\begin{small}
    \begin{align*}
  \mathbb{E}\left[|\mathsf{X_{faulty}}|\right] &\leq 2^n\left(\sum_{r=1}^n {N\binom{n}{r}} (2p-p^2)^{r}\left(1-\frac{1}{2}(1-p)^2\right)^L\right)^N. 
\end{align*}
\end{small}
From Markov inequality, $P(|\mathsf{X_{faulty}}|\geq 1)\leq \mathbb{E}\left[|\mathsf{X_{faulty}}|\right]$. Hence,  $ P(|\mathsf{X_{faulty}}|<1)$ is at least
\begin{align*}
   1- 2^n\left(\sum_{r=1}^n {N\binom{n}{r}} (2p-p^2)^{r}\left(1-\frac{1}{2}(1-p)^2\right)^L\right)^N. 
\end{align*}
Lastly, it can be verified that $P(|\mathsf{X_{faulty}}|<1)\geq 1-\epsilon_1$ if $ \log_{(1-\frac{1}{2}(1-p)^2)}\left(\frac{\sqrt[N]{\epsilon_1/2^n}}{N((2p-p^2+1)^n-1)}\right)<L=\beta n$. 
\end{proof}

\begin{corollary}
For given $n, \epsilon_1, p$, $\beta_{\mathsf{Th}}$  is minimum at $N = \ln\left(\frac{2^n}{\epsilon_1}\right)$.
\end{corollary}
\begin{proof}

Observe that $$
\underset{N\in\mathbb{Z}_{+}}{\argmin} ~\beta_{\mathsf{Th}} = \underset{N\in\mathbb{Z}_{+}}{\argmin}~\log_{2}\left(\frac{N((1+2p-p^2)^n-1)}{\sqrt[N]{\epsilon_1/2^n}}\right),
$$
which can be verified to be maximum at $N = \ln\left(\frac{2^n}{\epsilon_1}\right)$.
\end{proof}

\begin{figure}[H]
    \centering
    \includegraphics[width=0.5\textwidth]{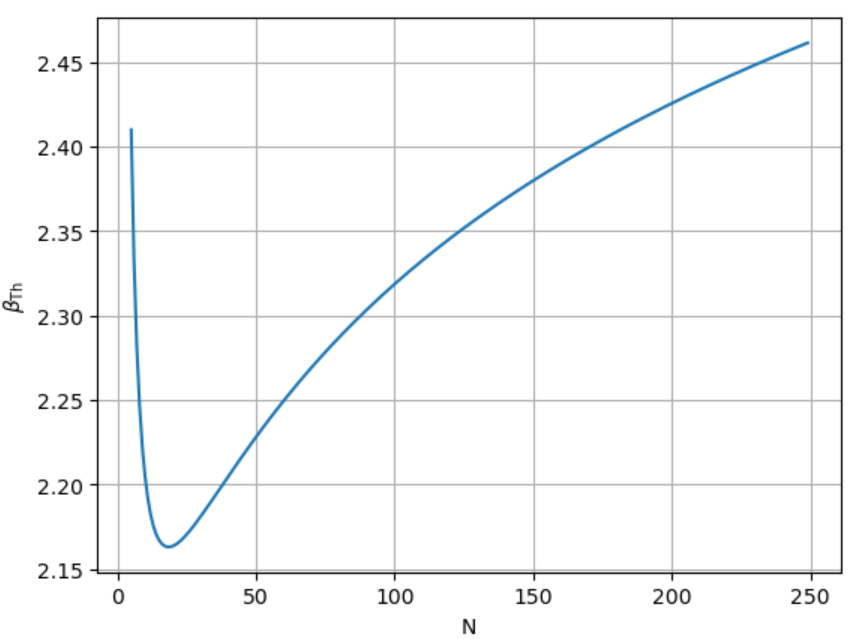}
    \label{fig:beta_N}
    \caption{Plot of $\beta_{\mathsf{Th}}$ versus $N$ for $n = 20, p = 0.3, \epsilon_1 = 0.01$.}
\end{figure}

\begin{corollary}\label{b_th_n}
For given $N, \epsilon_1, p$, $\beta_{\mathsf{Th}}= \cO(1)$. Specifically, 
\begin{align*}
\beta_{\mathsf{Th}} < \frac{\log_2(N(1+2p-p^2)) + \frac{1-\log_2(\epsilon_1)}{N}}{1-\log_2(1+2p-p^2)}.
\end{align*}
\end{corollary}
\begin{proof}
From Lemma \ref{b_th}, $\beta_{\mathsf{Th}} = \frac{\log_{2}\left(\frac{N((1+2p-p^2)^n-1)}{\sqrt[N]{\epsilon_1/2^n}}\right)}{n(1-\log_2(1+2p-p^2))}$. Therefore, we have that 
\begin{align*}
    \beta_{\mathsf{Th}} &< \frac{\log_{2}\left(\frac{N((1+2p-p^2)^n)}{\sqrt[N]{\epsilon_1/2^n}}\right)}{n(1-\log_2(1+2p-p^2))}\\
    &= \frac{\log_{2}(N) + n \log\left(1+2p-p^2\right)-\left(\frac{\log(\epsilon_1)-n}{N}\right)}{n(1-\log_2(1+2p-p^2))} \\
    &=  \frac{\log_2(1+2p-p^2)+\frac{1}{N}}{(1-\log_2(1+2p-p^2)} + \frac{\log_2(N)-\frac{\log_2(\epsilon_1)}{N}}{n(1-\log_2(1+2p-p^2)}  \\
    &< \frac{\log_2(N(1+2p-p^2)) + \frac{1-\log_2(\epsilon_1)}{N}}{1-\log_2(1+2p-p^2)}.
\end{align*}
\end{proof}

\begin{definition}
    Given a partitioning $\cP=\{P_1,P_2, \ldots, P_M\}$, we define a \textbf{labelling}, denoted by $\cL$, as a length-$M$ vector of distinct addresses from $C$ such that $\cL[i] \in \{\bfx: \forall (\bfy,\bfd')\in P_i, P(\bfx|\bfy)>0, \}$, where $\cL[i]$ denotes the $i$-th element of $\cL$, and $i\in[M]$.  
\end{definition}

We denote the set of all possible labellings for a given partitioning $\cP$ by $\mathbb{L}_{\cP,\cY}$. Given the true partitioning $\cP^*$, we define the \textit{true labelling}, denoted by $\cL^*$, as the labelling in which for each partition $\cS_N((\bfx_i,\bfd_i))$, the assigned label is $\bfx_i$, where $i\in[M]$. Note that if $\cP\neq\cP^*$ then $\cL^*\notin\mathbb{L}_{\cP,\cY}$. Further, if $|\mathbb{L}_{\cP^*,\cY}| = 1$ then $\mathbb{L}_{\cP^*,\cY} = \{\cL^*\}$. Let $\cG'' = (\cX,E'')$, where $\cX = C$. There is a directed edge $\bfx\rightarrow\tx$ if all of the $N$ channel outputs of $\bfx$ are erased at the positions where $\bfx$ and $\tx$ differ, i.e., $\{\tx\} \in \{\bigcap_{(\bfy,\bfd')\in\cS_N(\bfx,\bfd)}E_{(\bfy,\bfd')}\}$. 

\begin{proposition}
$|\mathbb{L}_{\cP^*,\cY}| = 1$ if and only if there are no directed cycles in $\cG''$.    
\end{proposition}

\begin{figure}[H]
    \centering
\includegraphics[width=0.5\textwidth]{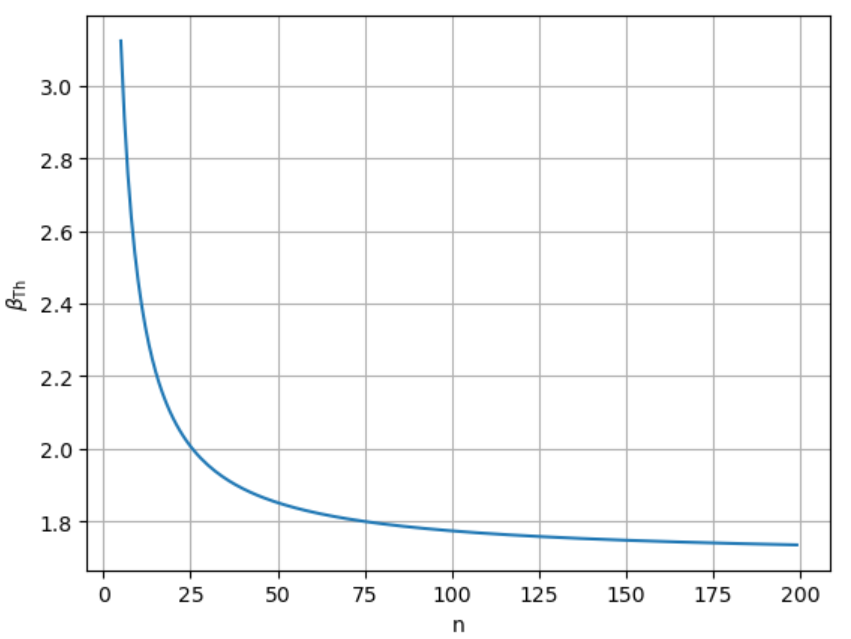}
    \label{fig:beta_n}
    \caption{Plot of $\beta_{\mathsf{Th}}$ versus $n$ for $N = 2, p = 0.2, \epsilon_1 = 0.01$.}
\end{figure}

In the next lemma, we derive  a threshold on $N$ such that for $N\geq N_\mathsf{Th}, \mathbb{L}_{\cP^*,\cY} = \{\cL^*\}$ with probability at least $1-\epsilon_2$.
\begin{restatable}{lemma}{BoundOnN}\label{N_th}
     For $N\geq N_\mathsf{Th} \triangleq\frac{\log_2\left(\sqrt[n]{\frac{\epsilon_2+2^{n}}{2^{n}}}-1\right)}{\log_2(p)}$, 
     we have that $\mathbb{L}_{\cP^*,\cY} = \{\cL^*\}$ with probability at least $1-\epsilon_2$.
\end{restatable}
\begin{proof}    
Let $\mathsf{X_{faulty}}$ denote the set of nodes in $\cG''$ that have at least one outgoing edge. For $\bfx,\tx \in \cX$, let $r=d_H(\bfx,\tx)$. Note that $\bfx\rightarrow\tx$ happens with probability $p^{rN}$. Therefore, the probability that $\bfx$ has no outgoing edges is $\prod_{r=1}^{n} (1-p^{rN})^{\binom{n}{r}}$. Hence, from linearity of expectation,
\begin{align*}
  \mathbb{E}\left[|\mathsf{X_{faulty}}|\right] &= \sum_{\bfx\in \cX}\mathbb{E}[\mathbb{I}_{\{\bfx\in\mathsf{X_{faulty}}\}}]\\
  &= 2^n\left(1-\prod_{r=1}^{n}\left(1-p^{rN}\right)^{\binom{n}{r}}\right).
  \intertext{From Weierstrass inequality, we have that  $\prod_{r=1}^{n}\left(1-p^{2rN}\right)^{\binom{n}{r}}\geq 1-\sum_{r=1}^{n} {n \choose r}p^{2rN}$. Therefore,}
  \mathbb{E}\left[|\mathsf{X_{faulty}}|\right] &\leq  2^n\left(\sum_{r=1}^{n} {n \choose r}p^{rN}\right)= 2^n\left((1+p^N)^n-1\right). 
\end{align*}
From Markov inequality, $P(|\mathsf{X_{faulty}}|\geq 1)\leq \mathbb{E}\left[|\mathsf{X_{faulty}}|\right]$. Hence,  
\begin{align*}
    P(|\mathsf{X_{faulty}}|<1)\geq 1- 2^n\left((1+p^N)^n-1\right). 
\end{align*}
Lastly, it can be verified that $P(|\mathsf{X_{faulty}}|<1)\geq 1-\epsilon_2$ if $N > \log_p (\sqrt[n]{\frac{\epsilon_2+2^{n}}{2^{n}}}-1)$. 

\end{proof}

\begin{corollary}\label{N_Th_n}
For given, $p, \epsilon_2$, $N_{\mathsf{Th}} = \Theta(n)$. Specifically,  
\begin{align*}
    \frac{n+\log_2\left(\frac{n}{\epsilon_2\ln(2)}\right)}{\log_2\left(\frac{1}{p}\right)}>N_\mathsf{Th}&>\frac{n+\log_2\left(\frac{n}{\epsilon_2}\right)}{\log_2\left(\frac{1}{p}\right)}.
\end{align*}
\end{corollary}
\begin{proof}
    From Lemma \ref{N_th}, $N_\mathsf{Th} =  \frac{\log_2\left(\sqrt[n]{1+\frac{\epsilon_2}{2^{n}}}-1\right)}{\log_2(p)}$. For $0<a<1, 0<b<1$, it can be verifed that $1+ab\ln(2)<2^{ab}<(1 + a)^b<1+ab$. Therefore, we have that 
    \begin{align*}
        \frac{\log_2\left(\frac{\epsilon_2\ln(2)}{n2^n}\right)}{\log_2(p)} > N_\mathsf{Th}&>\frac{\log_2\left(\frac{\epsilon_2}{n2^n}\right)}{\log_2(p)}\\
        \iff
        \frac{n+\log_2\left(\frac{n}{\epsilon_2\ln(2)}\right)}{\log_2\left(\frac{1}{p}\right)}>N_\mathsf{Th}&>\frac{n+\log_2\left(\frac{n}{\epsilon_2}\right)}{\log_2\left(\frac{1}{p}\right)}.
    \end{align*}
\end{proof}
\begin{figure}[H]
    \centering
\includegraphics[width=0.5\textwidth]{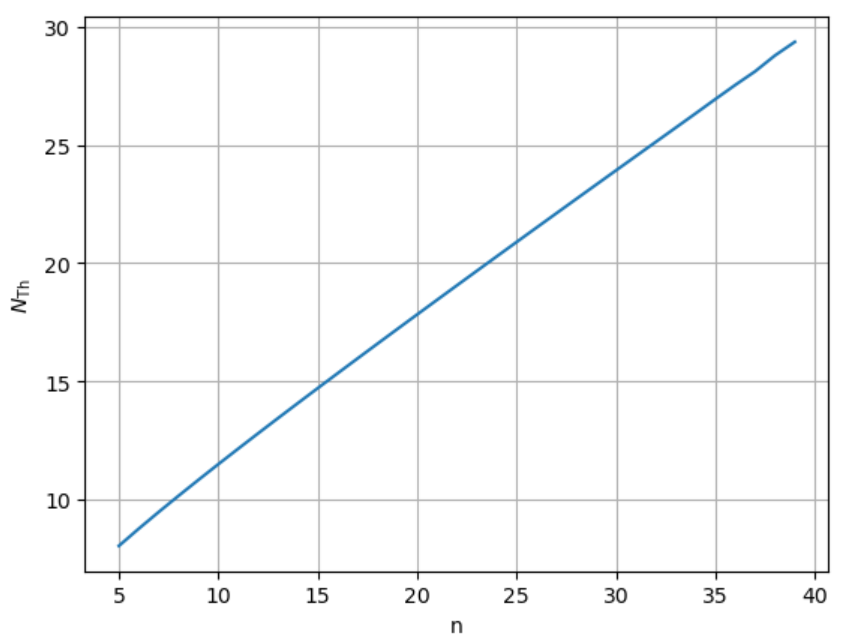}
    \label{fig:N_n}
    \caption{Plot of $N_{\mathsf{Th}}$ versus $n$ for $p = 0.3, \epsilon_2 = 0.01$.}
\end{figure}

Thus, we define the region $\cR$ as $\cR\triangleq\{(\beta,N):\beta\geq\beta_\mathsf{Th}, N\geq N_{\mathsf{Th}}\}$. In the next theorem, we give a sufficient condition for the existence of a unique $N$-permutation. 
\begin{restatable}{theorem}{TheOfUniquePermutation}\label{thm:prob1}
For $(\beta,N)\in\cR$, it is possible to identify the true permutation with probability at least $1-\epsilon$, when $\epsilon_1, \epsilon_2 < \frac{\epsilon}{2}$. 
\end{restatable}
\begin{proof}
    From Lemma~\ref{b_th} and~\ref{N_th}, it follows that for $\beta>\beta_\mathsf{Th}$ and $N>N_{\mathsf{Th}}$, $\mathbb{P}_\cY = \{\cP^*\}$ with probability $(1-\epsilon_1)$ and $\mathbb{L}_{\cP^*,\cY} = \{\cL^*\}$ with probability $(1-\epsilon_2)$, respectively. Hence, for $\beta>\beta_\mathsf{Th}$ and $N>N_{\mathsf{Th}}$, there exists only one valid permutation with probability $(1-\frac{\epsilon}{2})^2>(1-\epsilon)$. 
\end{proof}
\vspace{0.2cm}
From Corollaries \ref{b_th_n} and \ref{N_Th_n}, we observe that $\beta_{\mathsf{Th}}<\beta^*$ and $N_{\mathsf{Th}}<\nu^* n$ for some constants $\beta^*$ and $\nu^*$. 
This means that we only require data parts to be of length $L=\beta^* n$ and the number of reads to be $N=\nu^* n$ so that correct identification occurs with high probability.
In the next section, we design an algorithm to find the true permutation with a small number of data comparisons.

\section{Data-driven Pruning Algorithm}
\label{Sec:DPA}
As the receiver has access to the set of addresses, we design an algorithm that reduces the number of data comparisons by comparing a pair of reads if and only if they agree at the positions that are not erased in the address part. Hence, similar to the peeling matching algorithm, we first build the bipartite graph $\cG = (\cX\cup\cY, E)$ as described in Section~\ref{Sec:PMA}. Let $\cN_{(\bfy,\bfd')}$ denote the two-hop neighborhood of $(\bfy,\bfd')$ in $\cG$. Note that for $(\bfy,\bfd'), (\widetilde{\bfy},\widetilde{\bfd'})\in\cY$,  $(\bfy,\bfd')\in\cN_{(\widetilde{\bfy},\widetilde{\bfd'})}$ if and only if $\bfy\not\cong\widetilde{\bfy}$. In the next lemma, we calculate the expected value of $|\cN_{(\bfy,\bfd')}|$.

\begin{figure}[H]
\centering
\resizebox{6.5cm}{!}
{
\begin{tikzpicture}

\filldraw (1.5,2) circle (1pt);
\filldraw (1.5,2.5) circle (1pt);
\filldraw (1.5,3) circle (1pt);
\filldraw (1.5,3.5) circle (1pt);
\filldraw (1.5,4) circle (1pt);
\filldraw (1.5,4.5) circle (1pt);

\draw (2.5,2) node{\textcolor{blue}{$(\bfy_6,\bfd'_6)$}};
\draw (2.5,2.5) node{$(\bfy_5,\bfd'_5)$};
\draw (2.5,3) node{$(\bfy_4,\bfd'_4)$};
\draw (2.5,3.5) node{$(\bfy_3,\bfd'_3)$};
\draw (2.5,4) node{$(\bfy_2,\bfd'_2)$};
\draw (2.5,4.5) node{$(\bfy_1,\bfd'_1)$};

\filldraw (-1.5,2.5) circle (1pt);
\filldraw (-1.5,4) circle (1pt);

\draw [-stealth] (-1.5,2.5) -- (1.5,2);
\draw [-stealth] (-1.5,2.5) -- (1.5,2.5);
\draw [-stealth] (-1.5,2.5) -- (1.5,3);
\draw [-stealth] (-1.5,2.5) -- (1.5,3.5);

\draw [-stealth] (-1.5,4) -- (1.5,2.5);
\draw [-stealth] (-1.5,4) -- (1.5,3);
\draw [-stealth] (-1.5,4) -- (1.5,3.5);
\draw [-stealth] (-1.5,4) -- (1.5,4);
\draw [-stealth] (-1.5,4) -- (1.5,4.5);

\draw (-2,4) node{$\bfx_1$};
\draw (-2,2.5) node{$\bfx_2$};

\end{tikzpicture}
}
\caption{Let $N=3$. For $(\bfy_6,\bfd'_6)$, we can potentially identify the remaining $2$ copies by  performing only $|\cN_{(\bfy_6,\bfd'_6)}| = 3$ data comparisons.}
\label{fig:DPA}
\end{figure}
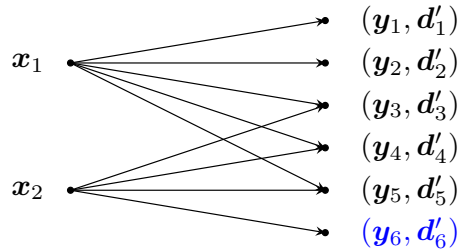

\begin{restatable} {lemma}{SizeOfTwoHop}
\label{2hop}
    For a given $(\bfy,\bfd') \in \cY$,
    \begin{align*}
        \mathbb{E}[|\cN_{(\bfy,\bfd')}|~|~(\bfy,\bfd')] = N2^r(1+p)^{n-r}-1,
    \end{align*}
where $r$ denotes the number of erasures in $\bfy$. Further, $ \mathbb{E}[|\cN_{(\bfy,\bfd')}|] = N(1+2p-p^2)^{n}-1$. 
\end{restatable}
\begin{proof}
    Let the number of erasures in $\bfy$ be $r$. Let $(\bfy,\bfd')\in\cS_N((\bfx,\bfd))$. For $(\tx,\tilde{\bfd})\in\cX$ that does not differ from $(\bfx,\bfd)$ at the non-erased positions in $(\bfy,\bfd')$, $\cS_N(((\tx,\tilde{\bfd})))\subset\cN_{(\bfy,\bfd')}$. Note that there are $2^r-1$ such $(\tx,\tilde{\bfd})$. 
    For $(\tx,\tilde{\bfd})\in\cX$ that differ from $(\bfx,\bfd)$ at $i$ out of the $n-r$ non-erased positions in $(\bfy,\bfd')$, we have that for $(\widetilde{\bfy},\tilde{\bfd'}) \in \cS_N((\tx,\tilde{\bfd}))$, $P((\widetilde{\bfy},\tilde{\bfd'})\in\cN_{(\bfy,\bfd')}) = p^i$. Note that there are $2^r\binom{n-r}{i}$ such $(\tx,\widetilde{\bfd})$. Hence, from linearity of expectations, 
    $\mathbb{E}[|\cN_{(\bfy,\bfd')}|~|~(\bfy,\bfd')] = (N(2^r-1)+ N-1) +N2^r\sum_{i=1}^{n-r}\binom{n-r}{i}p^i = N2^r(1+p)^{n-r}-1$. Using  law of total expectation, \begin{align*}
        \mathbb{E}[|\cN_{(\bfy,\bfd')}|] &= \sum_{r=0}^n \binom{n}{r}p^r(1-p)^{n-r}\left(N2^r(1+p)^{n-r}-1\right)\\
        &= N(1+2p-p^2)^{n}-1. 
    \end{align*} 
\end{proof}

The data-driven pruning algorithm as described below, iteratively selects the right node $(\bfy,\bfd')$ with the smallest two-hop neighborhood in $\cY$ and then as shown in Fig.~\ref{fig:DPA}, performs $|\cN_{(\bfy,\bfd')}|$ data comparisons to identify the remaining $N-1$ copies. Note that this pruning procedure finds the remaining $N-1$ copies if and only if  $(\bfy,\bfd')\not\in\mathsf{Y_{faulty}}$. Let $\cP_\cG = \left(\cX \cup \cY, \cP_E\right)$ denote the bipartite matching identified by the data-driven pruning algorithm.  
\begin{algorithm}
\caption{Data-driven Pruning Algorithm}
\label{Alg2}
\begin{algorithmic}[1]
\Procedure{Prune}{$\cG,(\widetilde{\bfy},\widetilde{\bfd'})$}
\State $(\widetilde{\bfy},\widetilde{\bfd'})\longrightarrow\mathsf{Pruned}$, $\cT = \{\}$
\For{$(\bfy,\bfd')\in  \cN_{(\widetilde{\bfy},\widetilde{\bfd'})}$}
    \If{$(\bfy,\bfd')\cong (\widetilde{\bfy},\widetilde{\bfd'})$}
    \State $(\bfy,\bfd') \longrightarrow \cT $
    \EndIf
\EndFor

\If{$|\cT| = N-1$}
\State Let $\cX^* = \bigcap_{(\bfy,\bfd')\in\cT}E_{(\bfy,\bfd')}$
\For{$(\bfy,\bfd')\in \cT$}    
    \State Remove $\{(\bfx, (\bfy,\bfd')):\bfx\notin \cX^*\}$ from $E$
    \State $(\bfy,\bfd')\longrightarrow\mathsf{Pruned}$
\EndFor
\EndIf
\EndProcedure
\vspace{0.2cm}

\Procedure{Pruning Algorithm}{$\cP_\cG, \cG$}
\State$\mathsf{Pruned}=\{\}$
\While{$|\mathsf{Pruned}|<N2^n$}
\State $(\widetilde{\bfy},\widetilde{\bfd'}) = \argmin\{|\cN_{(\bfy,\bfd')}|:(\bfy,\bfd')\in \cY\}$
\State PRUNE $(\cG,(\widetilde{\bfy},\widetilde{\bfd'}))$
\EndWhile
\State \Return $\mathsf{PMA}(\cP_\cG,\cG)$
\EndProcedure
\end{algorithmic}
\end{algorithm}

\begin{restatable}{proposition}{ProlgFindsPermutation}
\label{prop2}
For $(\beta,N)\in\cR$, Algorithm~\ref{Alg2} finds the true permutation with probability at least $1-\epsilon$, when $\epsilon_1, \epsilon_2 < \frac{\epsilon}{2}$. 
\end{restatable}
\begin{proof}
    For $\beta>\beta_{\mathsf{Th}}$, every left node has at least one non-faulty channel output with probability at least $(1-\epsilon_1)$. Thus, the data-driven pruning algorithm identifies the true partitioning with probability at least $(1-\epsilon_1)$. For $N>N_{\mathsf{Th}}$, there exists only one valid labelling, viz. the true labelling with probability at least $(1-\epsilon_2)$. Thus, the data-driven pruning algorithm identifies the true permutation with probability at least $(1-\frac{\epsilon}{2})^2>(1-\epsilon)$. 
\end{proof}
\section{Analysis of Data-driven Pruning Algorithm}
\label{Sec:DPA_analysis}
In this section, we analyse the expected number of data comparisons performed by Algorithm \ref{Alg2} for three subregions of $\cR$. 
In the next lemma, we give an upper bound on the expected number of data comparisons performed by Algorithm \ref{Alg2} when $(\beta,N)\in\cR$. 
\begin{restatable}{lemma}{ExepectedComForAlg}\label{ExpectedComparisons}
    The expected number of data comparisons performed by Algorithm \ref{Alg2} when $(\beta,N)\in\cR$ is at most 
    \begin{align*}
      \cU_0 \triangleq N^22^n\left(1+2p-p^2\right)^n.
    \end{align*}
\end{restatable}
\begin{proof}
  Note that the number of data comparisons performed by Algorithm \ref{Alg2} is at most $\sum_{(\bfy,\bfd')\in\cY}|\cN_{(\bfy,\bfd')}|$. From linearity of expectations, $\mathbb{E}\left[\sum_{(\bfy,\bfd')\in\cY}|\cN_{(\bfy,\bfd')}|\right] = \underset{(\bfy,\bfd')\in\cY}{\sum}\mathbb{E}\left[|\cN_{(\bfy,\bfd')}|\right]$. From Lemma \ref{2hop}, the result follows. 
\end{proof}

Let $\beta_0$ be a threshold on $\beta$ such that for $\beta\geq\beta_0$, $P(|\mathsf{Y_{faulty}}|>1) < \epsilon_1$. In the next lemma, we derive this threshold $\beta_0$.

\begin{restatable}{lemma}{BetaLargerThanBetaZero}
    
For $\beta\geq\beta_0\triangleq \frac{\log_2\left(\frac{\epsilon_1}{2^{n}N^2((1+2p-p^2)^n-1)}\right)}{n\log_2\left(1-\frac{1}{2}(1-p)^2\right)}$, $P(|\mathsf{Y_{faulty}}|>1) < \epsilon_1$. 
\end{restatable}
\begin{proof}
If $\mathbb{E}\left[|\mathsf{Y_{faulty}}|\right]< {\epsilon_1}$ then from Markov inequality the result follows. From Lemma \ref{right_faulty}, we have that $$
\mathbb{E}\left[|\mathsf{Y_{faulty}}|\right]=N2^n\left(1-\prod_{r=1}^{n} \left(1-(2p-p^2)^{r}\left(1-\frac{1}{2}(1-p)^2\right)^L\right)^{N\binom{n}{r}}\right).
$$
Hence, to show that $\mathbb{E}\left[|\mathsf{Y_{faulty}}|\right]< {\epsilon_1}$ it is sufficient to show that $$1-\frac{\epsilon_1}{2^{n}N}<\prod_{r=1}^{n} \left(1-(2p-p^2)^{r}\left(1-\frac{1}{2}(1-p)^2\right)^L\right)^{N\binom{n}{r}}. $$
Using Weierstrass inequality we have that \begin{align*}
&\prod_{r=1}^{n} \left(1-(2p-p^2)^{r}\left(1-\frac{1}{2}(1-p)^2\right)^L\right)^{N\binom{n}{r}}\\
&\geq 1-\sum_{r=1}^n N\binom{n}{r} (2p-p^2)^{r}\left(1-\frac{1}{2}(1-p)^2\right)^L\\
& = 1-N\left(\left(1-\frac{1}{2}(1-p)^2\right)^L \left((1+2p-p^2)^n-1\right)\right),
\end{align*}
and thus it is enough to show that
$$N\left(\left(1-\frac{1}{2}(1-p)^2\right)^L \left((1+2p-p^2)^n-1\right)\right)<\frac{\epsilon_1}{2^{n}N}.$$
Lastly, it can be verified that $\mathbb{E}\left[|\mathsf{Y_{faulty}}|\right]< {\epsilon_1}$ if $ \log_{(1-\frac{1}{2}(1-p)^2)}\left(\frac{\epsilon_1}{2^{n}N^2((1+2p-p^2)^n-1)}\right)<L=\beta n$. 

\end{proof}

We define $\cR'\subseteq\cR$ as $\cR'\triangleq\{(\beta,N):\beta\geq\beta_0, N\geq N_\mathsf{Th}\}$. To analyse the expected number of data comparisons performed by Algorithm~\ref{Alg2} when $(\beta,N)\in\cR'$, we define the notion of order of a left node.  

\begin{definition}
A node $\bfx\in\cX$ has order $s$ if $\min\{|E_{(\bfy,\bfd')}|: (\bfy,\bfd') \in \cS_N(\bfx,\bfd) \} = s$. 
\end{definition}

For $s\in[2^n]$, let $\Xlr$ denote the set of left nodes with order $s$. In the next lemma, we calculate the probability that a left node has order $s$. 
\begin{restatable}{lemma}{OrderOfLeftNodes}
    
For $\bfx\in \cX$, before the initiation of Algorithm~\ref{Alg2}, $ P(\bfx\in \Xlr)$ is
$$
\begin{cases}
       \left(\sum_{i=\ell}^{n}\binom{n}{i}p^i(1-p)^{n-i}\right)^N - \left(\sum_{i=\ell+1}^{n}\binom{n}{i}p^i(1-p)^{n-i}\right)^N & s \in \{2^\ell, \ell\in[0:n]\} \\
        0 & \text{otherwise}.
\end{cases}
$$
\end{restatable}
\begin{proof}
    Note that before any edges are removed from $\cG$, the degree of a right node can take values only from the set $\{2^\ell, \ell\in[0:n]\}$. Therefore, for $s \notin \{2^\ell, \ell\in[0:n]\}, P(x\in \Xlr) = 0$. For $s \in \{2^\ell, \ell\in[0:n]\}$, it can be verified that $P(\bfx\in \Xlr) = P (\bigcup_{(\bfy,\bfd') \in \cS_N(\bfx,\bfd)} |E_{(\bfy,\bfd')}|\geq s) - P (\bigcup_{(\bfy,\bfd') \in \cS_N(\bfx,\bfd)} |E_{(\bfy,\bfd')}|\geq s+1)$. Therefore, the result follows. 
\end{proof}

In the next lemma, we derive an upper bound on the expected number of data comparisons performed by Algorithm~\ref{Alg2} when $(\beta, N)\in\cR'$.
\begin{restatable}{lemma}{TotalExpLemma}\label{ExpectedComparisons2}
    For $(\beta, N)\in\cR'$, the expected number of data comparisons performed by Algorithm~\ref{Alg2} is at most
\begin{align*}
       \cU_1 \triangleq \sum_{r=0}^n \mathbb{E}[|\mathsf{X}_{2^r}|]N2^r((1+p)^{n-r}).
\end{align*}
 \end{restatable}   
 \begin{proof}
 Note that for $\beta>\beta_0$, there are no faulty right nodes with probability at least $1-\epsilon_1$. Hence, the expected number of data comparisons performed by Algorithm \ref{Alg2} to identify the $N$ channel outputs of $\bfx$ is at most $\mathbb{E}\left[\min\{|\cN_{(\bfy,\bfd')}|: (\bfy,\bfd') \in \cS_N((\bfx,\bfd))\}\right]$.
For $\bfx\in\cX$, by law of total expectation 
    \begin{align*}
        &\mathbb{E}\left[\min\{|\cN_{(\bfy,\bfd')}|: (\bfy,\bfd') \in \cS_N((\bfx,\bfd))\}\right] \\
        &= \mathbb{E}\left[\mathbb{E}\left[\min\{|\cN_{(\bfy,\bfd')}|: (\bfy,\bfd') \in \cS_N((\bfx,\bfd))\}\right]~|~x\in \Xlr\right]
        \intertext{From Lemma 5, $ \mathbb{E}[|\cN_{(\bfy,\bfd')}|~|~(\bfy,\bfd')] = N2^r(1+p)^{n-r}-1$, }
        &= \sum_{r=0}^{n}P(x\in\mathsf{X}_{2^r}) (N2^r(1+p)^{n-r}-1). 
    \end{align*}
From linearity of expectation, the result follows. 
\end{proof}

\begin{figure}[H]
    \centering
    \includegraphics[width=0.7\textwidth]{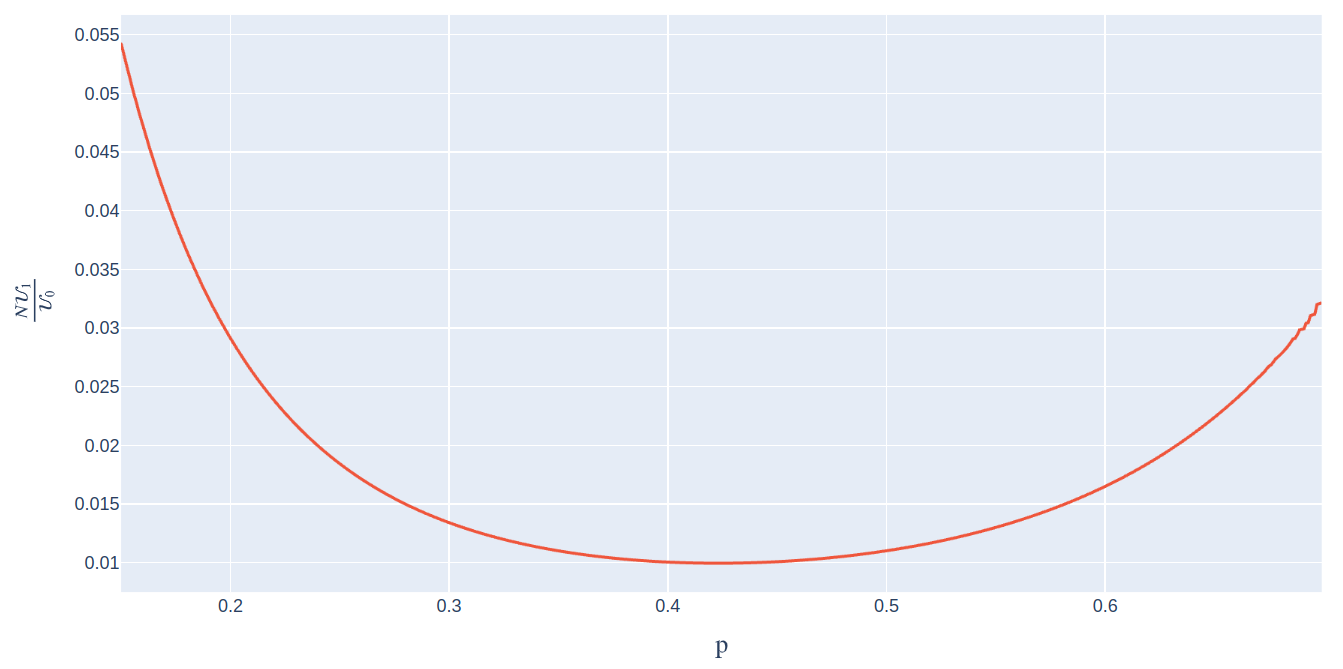}
    \label{fig:n30}
    \caption{Ratio of $N\cU_1$ to $\cU_0$ versus $p$ for $n=30$.}
\end{figure}

We now define the notion of confusability for  left nodes.  
\begin{definition}
\label{confusable}
Let $\bfx,\tx \in \cX$ then $\bfx$ is \textbf{confusable} with $\tx$, denoted by $\bfx\rightarrow\tx$, if there exists at least one $(\widetilde{\bfy},\widetilde{\bfd'})\in\cS_N((\tx,\widetilde{\bfd'}))$ such that $E_{(\widetilde{\bfy},\widetilde{\bfd'})}= \{\bfx,\tx\}$. 
\end{definition}

Next, we build a graph of left nodes, $T = (\cX, \mathsf{E_{conf}})$. Let $\bfx,\tx,\bfx'\in\cX$. Note that before the initiation of Algorithm~\ref{Alg2},  for $\bfx\rightarrow\tx$, it must be that $d_H(\bfx,\tx)=1$. For ease of analysis, we do not consider the confusable edges that would be generated over the course of Algorithm~\ref{Alg2}.  Thus, there is an edge $\bfx\rightarrow\tx \in \mathsf{E_{conf}}$  if and only if $\bfx$ is confusable with $\tx$ before the initiation of the algorithm. In the next lemma, we derive the probability that $\bfx$ has edges to all nodes in $S\subseteq\{\bfx':d_H(\bfx,\bfx')=1\}$.

\begin{restatable}{lemma}{LemmaTen}\label{lb1conf}
Let $\bfx\in\cX$ and let $S\subseteq\{\bfx':d_H(\bfx,\bfx')=1\}$. Then, 
\begin{align*}
P\left(\bigcup_{j=1}^{|S|}(\bfx\rightarrow\bfx_i)\right) &= 
\prod_{j=1}^{|S|}\left(1-\left(1-p(1-p)^{n-1}\right)^{N-j+1}\right),
\end{align*}
where $\bfx_i\in S$ for $i\in[|S|]$.
\end{restatable}
\begin{proof}
 From Definition~\ref{confusable}, $\bfx\rightarrow\tx$ there must exist a $\bfy\in \cS_N((\bfx,\bfd))$ such that $ E_{(\bfy,\bfd')}= \{\bfx,\tx\}$, which happens if and only if $\bfy$ is erased only at the position where $\bfx$ and $\tx$ differ, which happens with the probability $p(1-p)^{n-1}$. Therefore, $P\left(\bigcup_{j=1}^{|S|}(\bfx\rightarrow\bfx_i)\right) = 
\prod_{j=1}^{|S|}\left(1-\left(1-p(1-p)^{n-1}\right)^{N-j+1}\right)$. 
\end{proof}

Next, let $G_A = (\cX,\cE)$ be a directed $n$-cube~\cite{rgt}. A vertex $\bfx\in \cX$ has outgoing edges to the vertices $\{\bfx':d_H(\bfx,\bfx')=1, \bfx'\in \cX\}$. Let $G_A(p_e)$ denote a random sub-graph of $G_A$ where every edge in $\cE$ is selected with probability $p_e$. 

\begin{restatable}{proposition}{ApperenceConnectedCom}
\label{probappear}
The probability of the appearance of a connected component is greater in $T$ than in $G_A(p_T)$, where $p_T\triangleq \left(1-\left(1-p(1-p)^{n-1}\right)^{N-n+1}\right)$. 
\end{restatable}
\begin{proof}
    Since the probability that there is an edge $\bfx\rightarrow\bfx'$ is independent of the existence of the edge $\tx\rightarrow \bfx'$,  the proposition follows form Lemma~\ref{lb1conf}. 
\end{proof}

\begin{restatable}{lemma}{TAlmostConnected}   
For $N > N_0 \triangleq n-\frac{1}{\log(1-p(1-p)^{n-1})} = \cO_p\left(\frac{1}{p(1-p)^{n-1}}\right)$, $T$ is almost surely connected. 
\end{restatable}
\begin{proof}
From \cite{rgt}, we know that $G_A(p_e)$ is almost surely connected if  $p_e >\frac{1}{2}$. It can be verified that for $N >n-\frac{1}{\log(1-p(1-p)^{n-1})}$, $p_e = p_T >\frac{1}{2}$. Then from Proposition~\ref{probappear}, the result follows. 
\end{proof}
\vspace{0.2cm}
We define region $\cR''\subseteq\cR'$ as $\cR'' \triangleq \{(\beta,N):\beta\geq\beta_0, N\geq N_0\}$.    
\begin{restatable}{lemma}{LastLemma}\label{ExpectedComparisons3}
    The expected number of data comparisons performed by Algorithm \ref{Alg2} when $(\beta,N)\in\cR''$ is at most 
    \begin{align*}
      \cU_2 \triangleq N2^n\left(1+p\right)^n.
    \end{align*}
\end{restatable}
\begin{proof}
Since, the graph $T$ is connected, Algorithm \ref{Alg2} will always prune an order $1$ node. Hence, the result follows. 
\end{proof}

Hence, from Lemmas \ref{ExpectedComparisons}, \ref{ExpectedComparisons2} and \ref{ExpectedComparisons3}, the expected number of data comparisons performed by Algorithm \ref{Alg2} is only a $\kappa_{\beta,N}$-fraction of data comparisons required by clustering based approaches, where $\left(\frac{1+2p-p^2}{2}\right)^n \leq \kappa_{\beta,N} \leq \left(\frac{1+p}{2}\right)^n $.

\end{document}